%% file: main.tex
\documentclass[onefignum,onetabnum]{siamart190516}
\usepackage{marginnote}

\input{ex_shared}

\begin{document}
\maketitle

\begin{abstract}
    Regular games form a well-established class of games for analysis and synthesis of reactive systems. They include coloured Muller games, McNaughton games, Muller games, Rabin games, and Streett games. 
    These games are played on directed graphs $\mathcal G$ where Player 0 and Player 1 play by generating an infinite path $\rho$ through the graph. The winner is determined by specifications put on the set $X$ of vertices in $\rho$ that occur infinitely often. 
    These games are determined, enabling  the partitioning of $\mathcal G$ into two sets $W_0$ and $W_1$ of winning positions for Player 0 and Player 1, respectively.   Numerous algorithms exist that decide specific instances of regular  games, e.g., Muller games, by computing $W_0$ and $W_1$. 
    In this paper we aim to find general principles for designing uniform algorithms that decide all regular games.
    For this we utilise various recursive and dynamic programming algorithms 
    that leverage standard notions such as subgames and traps. 
    Importantly, we show that our techniques  improve or match the performances of existing algorithms for many instances of regular games.  
\end{abstract}

\noindent
\begin{keywords}
Regular games, coloured Muller games, Rabin games, deciding games.
\end{keywords}

\setcounter{page}{0}

\section{Introduction}

In the area of verification of reactive systems, 
studying games played on finite graphs is a key research topic \cite{gradel2002automata}. The recent work \cite{fijalkow2023games} serves as an excellent reference for the state-of-the-art in this area. Interest in these games arises from their role in modeling and verifying reactive systems  as games on graphs. Coloured Muller games, Rabin games, Streett games, Muller games, and McNaughton games  constitute  well-established classes of games for verification.   These games are played on finite bipartite graphs $\mathcal G$ between Player 0 (the controller) and Player 1 (the adversary, e.g., the environment). Player 0 and Player 1 play the game by producing an infinite path $\rho$ in $\mathcal G$.  Then the winner of
this play is determined by conditions put on $\mathsf{Inf}(\rho)$ the set  of all vertices in the path that appear infinitely often. 
Studying the algorithmic content of determinacy results for these games  is at the core of the area.  

Next we provide some basic definitions used in the study of regular games. After that we discuss known algorithms and compare them 
with our findings.  

\subsection{Arenas, regular games, subarenas, and traps}

All games that we listed above are played in arenas:
\begin{definition}
    An {\bf arena} $\mathcal{A}$ 
    is a bipartite  directed graph  $(V_0, V_1, E)$, where 
\begin{enumerate}
    \item   $V_0\cap V_1=\emptyset$, and $V=V_0\cup V_1$ is the set of nodes, also called {\bf positions}.
    \item $E\subseteq V_0 \times V_1 \cup V_1 \times V_0$ is the edge set where each node has an outgoing edge. 
    \item $V_0$ and $V_1$ are sets of positions for Player 0 and Player 1, respectively.  
\end{enumerate}
\end{definition}

 Players play the game in a given arena $\mathcal A$ by taking turns and moving a token along the edges of the arena. Initially, the token is placed on a node $v_0 \in V$. 
If $v_0\in V_0$, then Player 0 moves first. If $v_0\in V_1$, then Player 1 moves first.   
In each round of play, if the token is positioned 
on a Player $\sigma$'s position $v$, then Player $\sigma$ chooses $u\in E(v)$, moves the token to $u$ along the edge $(v,u)$, and the play continues on to the next round.
Note that condition $2$ on the arena guarantees that the players can always make a move at any round of the play. 

\begin{definition}
A {\bf play}, in a given arena $\mathcal A$, starting at  $v_0$, is an infinite sequence $\rho=v_0,v_1, v_2, \ldots$ such that $v_{i+1}\in E(v_i)$ for all $i\in \mathbb{N}$.  \ 
\end{definition}

Given a play  $\rho=v_0, v_1, \ldots$, the set  
$\mathsf{Inf}(\rho)=\{v\in V \mid \exists^{\omega} i (v_i=v)\}$ 
is called the {\bf infinity set} of $\rho$. The winner of this play  is determined by  a condition put on $\mathsf{Inf}(\rho)$. We list several of these conditions that are well-established in the area.

\begin{definition} \label{dfn:WinCon}
The following games played on a given arena $\mathcal{A}=(V_0,V_1,E)$ will be called {\bf regular games}:  
\begin{enumerate}

\item 
A {\bf coloured Muller game} is $\mathcal{G}=(\mathcal{A}, c, (\mathcal{F}_0, \mathcal{F}_1))$,  where $c: V\rightarrow C$ is a mapping from $V$ into the set $C$ of colors, $\mathcal{F}_0\cup \mathcal{F}_1=2^{C}$ and $\mathcal{F}_0\cap \mathcal{F}_1=\emptyset$. The sets  $\mathcal{F}_0$ and $\mathcal{F}_1$ are called {\bf winning conditions}. Player $\sigma$ {\bf wins} the play $\rho$
if $c( \mathsf{Inf}(\rho)) \in \mathcal{F}_\sigma$, where $\sigma=0,1$.

\item A {\bf McNaughton game} is the tuple $\mathcal{G}=(\mathcal{A}, W, (\mathcal{F}_0, \mathcal{F}_1))$, where 
$W\subseteq V$,  $\mathcal{F}_0\cup \mathcal{F}_1=2^{W}$ and $\mathcal{F}_0\cap \mathcal{F}_1=\emptyset$. Player $\sigma$ {\bf wins} the play $\rho$ if $\mathsf{Inf}(\rho)\cap W\in \mathcal{F}_\sigma$.

\item A {\bf Muller game} is the tuple $\mathcal{G}=(\mathcal{A}, (\mathcal{F}_0, \mathcal{F}_1))$, where  $\mathcal{F}_0\cup \mathcal{F}_1=2^{V}$ and $\mathcal{F}_0\cap \mathcal{F}_1=\emptyset$.   Player $\sigma$ {\bf wins} the play $\rho$
if $\mathsf{Inf}(\rho) \in \mathcal{F}_\sigma$. 

\item A {\bf Rabin game} is the tuple $\mathcal{G}=(\mathcal{A},(U_1, V_1), \ldots, (U_k, V_k))$, where 
$U_i, V_i \subseteq V$, $(U_i, V_i)$ is a {\bf winning pair}, and $k\geq 0$ is the {\bf index}. Player 0 {\bf wins} the play 
$\rho$ if there is a pair $(U_i, V_i)$ such that 
$\mathsf{Inf}(\rho) \cap U_i\neq \emptyset$ and $\mathsf{Inf}(\rho) \cap V_i= \emptyset$. Else, Player 1 wins.
\item A {\bf Streett game} is the tuple $\mathcal{G}=(\mathcal{A},(U_1, V_1), \ldots, (U_k, V_k))$, where 
$U_i$, $V_i$ are as in Rabin game. Player 0 {\bf wins} the play 
$\rho$ if for all $i\in \{1, \ldots, k\}$
if $\mathsf{Inf}(\rho) \cap U_i\neq \emptyset$ then $\mathsf{Inf}(\rho) \cap V_i\neq \emptyset$. Otherwise, Player 1 wins.

\item A {\bf KL game} is the tuple $\mathcal{G}=(\mathcal{A}, (u_1, S_1),\ldots,(u_t,S_t))$, where $u_i\in V$, $S_i\subseteq V$, \ $(u_i, S_i)$ is a {\bf winning pair}, and the {\bf index $t\geq 0$} is an integer. Player 0 {\bf wins} the play $\rho$ if there is a pair $(u_i, S_i)$ such that 
$u_i\in \mathsf{Inf}(\rho)$ and $\mathsf{Inf}(\rho) \subseteq S_i$. Else, Player 1 wins. 
\end{enumerate} 
\end{definition}

Note that the first three games are symmetric. 
Rabin games can be considered as dual to Streett games. 
The first five winning conditions are well-established conditions. The last condition is new. The motivation behind this new winning condition lies in the transformation of Rabin and Streett games into Muller games via the KL winning condition. In a precise sense, as will be seen in Section \ref{SS:Applications} via Lemma \ref{L:R-to-KL}, the KL condition serves as a compressed Rabin winning condition. The games  that we defined have natural parameters: 

\begin{definition}
The sequence $|C|$, $k$, $|W|$, $t$ is
the list of {\bf parameters} for   colored Muller games, Rabin  and Streett games, McNaughton games, and KL games, respectively. 
\end{definition} 

Since 
the  parameters $|W|$ and $|C|$ range in the interval $[0,|V|]$, we can call them {\bf small parameters}. The  parameters $k$ and $t$ range in $[0, 4^{|V|}]$ and $[0, 2^{|V|}\cdot |V|]$, respectively. Hence, we call them  (potentially) {\bf large parameters}.


\begin{definition}
Let $\mathcal A$ be an arena.   A {\bf pseudo-arena} of $\mathcal{A}$ determined by $X$ is the tuple $\mathcal{A}(X)=(X_0,X_1,E_X)$ where $X_0=V_0\cap X$, $X_1=V_1\cap X$, $E_X=E\cap (X\times X)$. If this pseudo-arena is an arena, then we call it the {\bf subarena} determined by $X$.
\end{definition}

The opponent of Player $\sigma$, where  $\sigma \in \{0,1\}$, 
is denoted by Player $\bar{\sigma}$. Traps are subarenas in games where one of the players has no choice but stay: 
 
\begin{definition}[$\sigma$-trap]
  A subarena $\mathcal{A}(X)$ is a {\bf $\sigma$-trap} for Player $\sigma$ if each of the following two conditions are satisfied:  (1)  For all $x\in X_{\bar\sigma}$ there is a  $y\in X_{\sigma}$ such that $(x,y) \in E$.  (2)  For all $x \in X_{\sigma}$ it is the case that $E(x)\subseteq X$.
  \end{definition}
 If $\mathcal{A}(X)$ is a $\sigma$-trap, then Player $\bar{\sigma}$ can stay in $\mathcal{A}(X)$ forever 
  if the player wishes to do so. 
  



Let $T$ be a subset of the arena  $\mathcal{A}=(V_0,V_1,E)$. The attractor of Player $\sigma$ to the set $T\subseteq V$, denoted $Attr_\sigma(T,\mathcal{A})$, is the set of positions from where Player $\sigma$ can force the plays into $T$. The attractor $Attr_\sigma(T,\mathcal{A})$ is computed as follows:

\begin{center}{
$W_0=T$, \ \ 
$W_{i+1}=W_i\cup \{u\in V_\sigma\mid E(u)\cap W_i \ne \emptyset\} \cup \{u\in V_{\bar\sigma} \mid E(u)\subseteq W_i\}$, \\\ and then set \  
$Attr_\sigma(T,\mathcal{A})=\bigcup_{i\geq 0}W_i$.}
\end{center}

The set $Attr_\sigma(T,\mathcal{A})$ can be computed in $O(|E|)$. We call $Attr_\sigma$ the attractor operator. Note that the set $V \setminus Attr_\sigma(T,\mathcal{A})$, the complement of the $\sigma$-attractor of $T$, is a $\sigma$-trap for all $T$. This set is the emptyset if and only if $V=Attr_\sigma(T,\mathcal{A})$.

\smallskip




A strategy for Player $\sigma$ is a function that receives as input initial segments of plays $v_0,v_1,\ldots, v_k$ where $v_k\in V_\sigma$ and outputs some $v_{k+1}$ such that $v_{k+1}\in E(v_k)$. An important class of strategies are finite state strategies.  
R. McNaughton  in \cite{mcnaughton1993infinite} proved that the winner in McNaughton games has a finite state winning strategy. 
W. Zielonka proves 
that the winners of regular games have finite state winning strategies \cite{zielonka1998infinite}.  S. Dziembowski, M. Jurdzinski, and I. Walukiewicz in \cite{dziembowski1997much} investigate the memory 
needed for the winners of coloured Muller games. They show that the memory  $|V|!$ is a sharp bound for finite state winning strategies.

In the study of games, the focus is placed on  solving  them. Solving a given regular game entails two key objectives. First, one aims to devise an algorithm that, when provided with a regular game $\mathcal G$, partitions the set $V$ into two sets $Win_0$ and $Win_1$ such that $v\in Win_{\sigma}$ if and only if Player $\sigma$ wins the game starting at $v$, where $\sigma \in \{0,1\}$. This is called the {\bf decision problem} where one wants to find out the winner of the game. Second,  one would like to design an algorithm that, given a regular game, extracts a winning strategy for the victorious player. This is known as the {\bf synthesis problem}.

Traditionally, research on regular games specifically selects an instance of regular games,  e.g., Muller games,  Rabin games or Streett games, and  studies the decision and synthesis problems for these instances. This paper however, 
instead of focusing on instances of regular games,  aims at finding uniform algorithms and general principles for deciding all regular games. Importantly, we show that our techniques based on general principles  improve or match the performances of existing decision algorithms for many instances of regular games.




\subsection{Our contributions in light of 
known algorithms} \label{S:Known-algorithms}

We provide two types of algorithms for deciding regular games. The first type are recursion based, and the second type are dynamic programming based. Recursive algorithms have been exploited in the area significantly.  To the best of our knowledge, dynamic programming techniques have not been much used in the area.  We utilise these techniques and improve  known algorithms for deciding all regular games defined above.  

The performances of algorithms for regular games $\mathcal G$ can be measured in two ways in terms of input sizes. One is when the input sizes are defined  as $|V|+|E|$.  The other is when the games $\mathcal G$ are presented {\em explicitely} that consists of listing $V$, $E$, and the corresponding winning conditions. In these explicit representations the sizes of  Muller, McNaughton, and coloured Muller games are bounded by $|V|+|E| +2^{|V|}\cdot |V|$. The sizes of 
Rabin and Streett games are bounded by $|V|+|E|+ 4^{|V|}\cdot |V|$. The sizes of the KL games are bounded by $|V|+|E| +2^{|V|}\cdot |V|^2$.  We use the notation $|\mathcal G|$ for these representations of games $\mathcal G$.   These two ways of representing inputs,  together with the {\em small } parameters $|C|$ and  $|W|$ and potentially {\em large} parameters $k$ and $t$,  should be taken into account in our discussion below.

\smallskip

{\bf 1: Coloured Muller games}. 
The folklore algorithms that decide coloured Muller games use induction on cardinality of the color set $C$ \cite{fijalkow2023games}. These algorithms are recursive and run in time $O(|C||E|(|C||V|)^{|C|-1})$ and space $O(|\mathcal{G}|+|C||V|)$.  Using the breakthrough quasi-polynomial time algorithm for parity games, C. Calude, S. Jain, B. Khoussainov, W. Li, and F. Stephan improve all the known algorithms for coloured Muller games with the running time $O(|C|^{5|C|}\cdot |V|^5)$ and space $O((|C|!|V|)^{O(1)})$ \cite{calude2017deciding}. \ Bj\"orklund et al. in  \cite{bjorklund2003fixed} showed that under the ETH it is impossible to decide coloured Muller games in time $O(2^{o(|C|)} \cdot |V|^a)$ 
for any $a>0$. C. Calude, S. Jain, B. Khoussainov, W. Li, and F. Stephan in \cite{calude2017deciding} improved this by showing that  under the ETH it is impossible to decide coloured Muller games in $2^{o(|C|\cdot \log(|C|))}Poly(|V|)$.
Their proof, however, implies that this impossibility result holds when $|C|\leq \sqrt{|V|}$. 
The table below now compares these results with  our algorithms.  

\begin{center}{
\begin{tabular}{|c| c|} 
 \hline
  Best known (running time, space) & Our algorithm (s) \\ 
 \hline
 ($O(|C|^{5|C|}\cdot |V|^5)$, $O((|C|!|V|)^{O(1)})$) & ($O(2^{|V|}|C||E|)$, $O(|\mathcal{G}|+2^{|V|}|V|)$) \\ 
 \cite{calude2017deciding}& Theorem \ref{Thm: Algorithm 1} (DP)  \\
 &  \\
  ($O(|C||E|(|C||V|)^{|C|-1})$, $O(|\mathcal{G}|+|C||V|)$)  &  ($O(2^{|V|}|V||E|)$, $O(|\mathcal{G}|+2^{|V|})$)\\
  folklore, e.g., see\cite{fijalkow2023games}& Theorem \ref{Thm: Algorithm 2} (DP)\\ & \\
  
  & ($O(|C|!\binom{|V|}{|C|}|V||E|)$, $O(|\mathcal{G}|+|C||V|)$)\\
  & Theorem \ref{Thm: SolveCMG time and space} (recursion) \\
 \hline
\end{tabular}}
\end{center}

The algorithms from Theorems  \ref{Thm: Algorithm 1} and \ref{Thm: Algorithm 2} are dynamic programming (DP) algorithms. One can verify that if  $|V|/\log \log (n) \leq |C|$, for instance $|V|/a< C$ where $a>1$, then:
\begin{enumerate}
    \item Running times of both of these algorithms  are better than $O(|C|^{5|C|}\cdot |V|^5)$, 
    \item Moreover, when the value of $|C|$ are in the range of 
    $|V|/\log \log (n) \leq |C|$, then
    these running times are in $2^{o(|C|\cdot \log(|C|))}Poly(|V|)$. This refines and strengthens  the impossibility result that under the ETH no algorithm exists that decides coloured Muller games in $2^{o(|C|\cdot \log(|C|))}Poly(|V|)$ \cite{calude2017deciding}.
    \item The spaces of both of these algorithms are also better 
    than $O((|C|!|V|)^{O(1)})$. 
    \item All of the previously known algorithms have superexponential running times. Our algorithms run in exponential time.   
    \end{enumerate}

 For small parameters  
 such as $|C|\leq \log (|V|)$, the algorithms from \cite{calude2017deciding} and \cite{fijalkow2023games} outperform our algorithms. Note, however, that the condition  $|V|/\log \log (n) \leq |C|$, as stated in our second observation above, is  reasonable and practically feasible. For instance, our algorithms are better for any value of $|C|$ with $C\geq |V|/a $, where $a>1$. 
 Also, the running times of our algorithms are exponential thus matching the bound of the impossibility result of Bj\"orklund et al.  mentioned above. 
 
 Our recursive algorithm from Theorem \ref{Thm: SolveCMG time and space} is a recast of standard recursive algorithms. However, as shown in the table,
 our careful running time analysis implies that our recursive algorithms  has a better running time and it matches the  space bounds of the previously known recursive algorithms. 
 
\smallskip
\noindent
{\bf 2: Rabin and Streett games.} E. A. Emerson and C. S. Jutla show that the problem of  deciding Rabin games is   
NP complete \cite{emerson1988complexity, emerson1999complexity}.
Hence, deciding Streett games is co-NP complete.
 Horn's algorithm  
for deciding Rabin games has the running time  $O(k!|V|^{2k})$ \cite{horn2005streett}. N. Piterman and A. Pnuelli show that Rabin  and Streett games can be decided in time $O(|E| |V|^{k+1}kk!)$ and space $O(nk)$ \cite{piterman2006faster}. The work of N. Piterman and A. Pnuelli remained state-of-the-art for Rabin games until the quasi-polynomial breakthrough for parity games by  C. Calude, S. Jain, B. Khoussainov, W. Li, and F. Stephan \cite{calude2017deciding}. They gave a FPT algorithm for Rabin games on $k$ colors by converting it to a parity game and using the quasi-polynomial algorithm. A Rabin game with $n$ vertices, $m$ edges and $k$ colors, can be reduced to a parity game with $N=nk^2k!$ vertices, $M=nk^2k!m$ edges and $K=2k+1$ colors \cite{emerson1991tree} (We will use these values of $N$, $M$, and $K$ below). By combining the reduction from Rabin to parity games and the state-of-the-art algorithms for parity games \cite{daviaud2020strahler, dell2022smaller, fearnley2017ordered, jurdzinski2017succinct} in a ``space-efficient'' manner, see for instance Jurdzi{\'n}ski and Lazi{\'c} \cite{jurdzinski2017succinct}, one can solve Rabin games in time $O(\max\{MN^{2.38},2^{O(K\log K)}\})$, but in exponential space. 
On substitution of the values of $M$ and $N$, the algorithm of Jurdzi{\'n}ski and Lazi{\'c} would take time at least proportional to $m(nk^2k!)^{3.38}$.  However, observe that the parity game obtained from a Rabin game is such that the number of vertices $N$ is much larger than the number of colors $K$. This results in $K\in o(\log (N))$. For cases where the number of vertices of the resulting parity game is much larger than the number of priorities, say the number of colors $(2k+1)$ is $o(\log(N))$, Jurdzi{\'n}ski and Lazi{\'c} also give an analysis of their algorithm that would solve Rabin games in time $O(nmk!^{2+o(1)})$ . Closely matching this are the run times in the work of Fearnley et al. \cite{fearnley2017ordered} who provides, among other bounds, a quasi-bi-linear bound of $O(MN\mathfrak{a}(N)^{\log\log N})$, where $\mathfrak{a}$ is the inverse-Ackermann function. In either case above, this best-known algorithm has at least a $(k!)^{2+o(1)}$ dependence in its run time, and takes the space proportional to $(nk^2k!)\log (nk^2k!)$; this has a $k!$ dependence again. R. Majumdar et al. in \cite{majumdar2024rabin} recently provided an algorithm that decides Rabin games  in $\tilde{O}(|E||V|(k!)^{1+o(1)})$ time and $O(|V|k\log k\log |V|)$ space. This breaks through the $2+o(1)$ barrier.  A. Casares 
et al. have shown in \cite{casares2024simple} that under the ETH it is impossible to decide Rabin games in $2^{o(k\log k)}Poly(|V|)$. Just like for coloured Muller games, this impossibility result holds true when $k \leq \sqrt{|V|}$.  The next table compares these results with our findings.

\begin{center}{
\begin{tabular}{|c| c|} 
 \hline
 Best known (running time, space) & Our algorithm (s) \\  
 \hline
  ($O(|E| |V|^{k+1}kk!)$, $O(|\mathcal{G}|+k|V|)$) &($O((k|V|+2^{|V|}|E|)|V|)$, $O(|\mathcal{G}|+2^{|V|}|V|)$)  \\
  \cite{piterman2006faster} & Theorem \ref{Thm: Solve rabin and streett DP} (DP)\\
  & \\
  
  ($\tilde{O}(|E||V|(k!)^{1+o(1)})$, $O(|\mathcal{G}|+k|V|\log k\log|V| )$) & ($O(|V|!|V|(|E|+k|V|))$, $O(|\mathcal{G}|+|V|^2)$)  \\ 
  \cite{majumdar2024rabin}  &  Theorem \ref{Thm: Solve rabin and streett Recursive} (recursion) \\
 \hline
\end{tabular}}
\end{center}

We single out four key parts of both of our algorithms:
\begin{enumerate}
\item In terms of time,  both our dynamic and recursive algorithms outperform the known algorithms when the parameter $k$ ranges in 
$[|V|, 4^{|V|}]$. In particular, when $k$ is  
polynomial on $|V|$ (which is a practical consideration), then our algorithms have better running times. 
\item Just as for coloured 
Muller games we refine the impossibility result of A. Casares 
et al.  under the assumption of the ETH  \cite{casares2024simple}. Namely, when the parameter $k\geq |V|\log  |V|$, both of our algorithms run in $2^{o(k\log k)}Poly(|V|)$. We consider the condition $k\geq |V|\log  |V|$ as reasonable and practically feasible.  
\item Our DP algorithm from Theorem \ref{Thm: Solve rabin and streett DP} 
is the first exponential time algorithm that decides Rabin games. The previously known algorithms run in superexponential times.
\item When $k$ falls into the range $[|V|, 4^{|V|}]$, then the recursive algorithm from Theorem \ref{Thm: Solve rabin and streett Recursive} performs the best in terms of space against other algorithms. 
\end{enumerate}
 
When the values of $k$ fall within  the range $[2^{|V|}, 4^{|V|}]$, our dynamic algorithm from Theorem \ref{Thm: Solve rabin and streett DP} outperforms other algorithms both in terms of space and time.

If Player 0 wins Rabin games, then the player has a memoryless winning strategy  \cite{emerson1999complexity}. Hence, one might  suggest the following way of finding the winner. Enumerate all memoryless strategies
and select the winning one. Even when the arena is a sparse graph, e.g., positions  have a fixed bounded out-degree, this process does not lead to exponential running time as the opponent might have a 
winning strategy with a large memory. 

\smallskip



\noindent
{\bf 3: Muller games.} 
Nerode, Remmel, and Yakhnis were the first who designed a competitive algorithm  that decides Muller games \cite{nerode1996mcnaughton}. The running time of their algorithm is $O(|V|!\cdot|V||E|)$. W. Zielonka \cite{zielonka1998infinite} examines Muller games  through  Zielonka trees. The size of Zielonka tree is $O(2^{|V|})$ in the worst case. \ S. Dziembowski, M. Jurdzinski, and I. Walukiewicz in \cite{dziembowski1997much} show that deciding Muller games with Zielonka trees as part of the input is in $\text{NP}\cap \text{co-NP}$. 
D. Neider, R. Rabinovich, and M. Zimmermann reduce Muller games to safety games with $O((|V|!)^3)$ vertices; 
safety games can be solved in linear time \cite{neider2014down}.  F. Horn in \cite{horn2008explicit} provides the first polynomial time decision algorithm for explicitly given Muller games with 
running time $O(|V|\cdot |\mathcal{F}_0| \cdot (|V|+|\mathcal{F}_0|)^2)$.   F. Horn's correctness proof has a non-trivial flaw.  B. Khoussainov, Z. Liang, and M. Xiao in \cite{liang2023connectivity} provide a correct proof of Horn's algorithm through new methods and improve the running time to $O( |\mathcal{F}_0|\cdot(|V|+|\mathcal{F}_0|)\cdot |V_0|\log |V_0|  )$. 
All the known algorithms that we listed above are either recursive algorithms or reductions to other known classes of games. Our algorithm is a dynamic programming algorithm, and 
to the best of our knowledge, the first dynamic algorithm that solves Muller games. 
The table below compares the best of these results for Muller games, in terms of time and space, with our algorithm from this paper:

\begin{center}
\begin{tabular}{|c| c|} 
 \hline
 Best known (running time, space) & Our algorithm \\ 
 \hline
($O( |\mathcal{F}_0|\cdot(|V|+|\mathcal{F}_0|)\cdot |V_0|\log |V_0|  )$, $O(|\mathcal{G}|+|\mathcal{F}_0|(|V|+|\mathcal{F}_0|))$)&($O(2^{|V|}|V||E|)$,$O(|\mathcal{G}|+2^{|V|})$)   \\
\cite{liang2023connectivity} & Theorem \ref{Thm:Muller-DP} (DP)\\
 \hline
\end{tabular}
\end{center}

One can see that the algorithm from \cite{liang2023connectivity}, in terms of running time and space, is better than our algorithm when $|\mathcal F_0|\leq \sqrt{2^{|V|}}$. However, our algorithm becomes competitive (or better) than 
the algorithm in \cite{liang2023connectivity} when  $|\mathcal F_0|>  \sqrt{2^{|V|}}$. Also, note that by running our algorithm and the algorithm in \cite{liang2023connectivity} in parallel, we get the best performing  polynomial time algorithm that solves explicitly given Muller games. 

\smallskip
\noindent
{\bf 4: McNaughton games.} R. McNaughton \cite{mcnaughton1993infinite} provided the first algorithm that decides McNaughton games in time $O(a^{|W|}\cdot |W|! \cdot |V|^3)$, for a constant $a>1$. Nerode, Remmel, and Yakhnis in \cite{nerode1996mcnaughton} improved the bound to $O(|W|!|W||E|)$. 
A. Dawar and P. Hunter proved that
deciding McNaughton games is a PSPACE-complete problem \cite{hunter2008complexity}. The table below compares our algorithms with currently the best algorithm that runs in time $O(|W|!|W||E|)$:

\begin{center}{
\begin{tabular}{|c| c|} 
 \hline
 Best known (running time, space) & Our algorithm (s) \\ 
 \hline
  ($O(|W||E||W|!)$, $O(|\mathcal{G}|+Poly(|V|))$) & ($O(2^{|V|}|W||E|)$,$O(|\mathcal{G}|+2^{|V|}|V|)$)   \\
  \cite{nerode1996mcnaughton}& Theorem \ref{Thm: solve Mcnaughton game DP} (DP) \\
  & \\ 
   & ($O(2^{|V|}|V||E|)$,$O(|\mathcal{G}|+2^{|V|})$) \\
   & Theorem \ref{Thm: solve Mcnaughton game DP}  (DP)\\
 \hline
\end{tabular}}
\end{center}

It is not too hard to see that when the value of the parameter $|W|$
exceeds $|V|/\log \log (n)$, then our algorithm has asymptotically  better running time. In particular, when $|W|$ is greater than any constant fraction of $|V|$, that is, $|W|\geq |V|/a$ where $a>1$, then 
our algorithm outperforms the bound in \cite{nerode1996mcnaughton}.

 \smallskip
 
In addition to all of the above,  we make the following three comments: (1) Running any of the previously known algorithms for any of the instances of regular games  in parallel with any of our appropriately chosen algorithms yields an improved running upper bound; (2) Exponential (on $|V|$) bounds are unavoidable due to the ETH considerations (as we explained above). In fact, the sizes of the games $\mathcal G$ can be exponential on $|V|$; (3) Even though our algorithms provide competitive running times, their  possible
limitations are in the use of large spaces and the lack of clear dependence on the parameters. However, when the inputs have exponential size, then our algorithms require linear space 
on the sizes of the inputs.

\section{The notion of full win and characterization of winning regions}\label{S:SolveMG}

In this section we develop a few concepts and techniques used throughout the paper. 
We first define the notion of {\em full win}. This will be used in designing dynamic programming algorithms for deciding regular games. Then we provide Lemma \ref{L: coloured muller characterization} that characterizes winning regions. This lemma is used for designing recursive algorithms for solving regular games. The last result of this section is Lemma \ref{L: 3 partitioning new}. We call the lemma trichotomy lemma as it characterises three cases: (1) Player 0 fully wins the game, (2) Player 1 fully wins the game,  and (3) none of the players fully wins the game.  
This lemma will be the basis of our dynamic algorithms.

\begin{definition}\label{Dfn:FullWin}
If $Win_\sigma(\mathcal{G})=V$, then 
player $\sigma$ {\bf fully wins $\mathcal{G}$}. 
Else, the player {\bf does not fully win $\mathcal{G}$}. If $Win_\sigma(\mathcal{G})\ne V$ and $Win_{\bar{\sigma}}(\mathcal{G})\ne V$, then  {\bf no player fully wins $\mathcal{G}$}.
\end{definition}


We now provide two lemmas that characterize winning regions in coloured Muller games. 
Later we  algorithmically implement the lemmas and analyse them.  We start with the first lemma. The statement of the lemma and  its equivalent forms  have been known and used in various forms \cite{mcnaughton1993infinite} \cite{fijalkow2023games}. Later we will utilise the lemma in our recursive algorithms through their detailed exposition and analysis. 

\begin{lemma}\label{L: coloured muller characterization}
    Let $\sigma\in\{0,1\}$ such that $c(V)\in \mathcal{F}_\sigma$. Then we have the following:
    \begin{enumerate}
        \item  If for all $c'\in c(V)$, $Attr_\sigma(c^{-1}(c'), \mathcal{A})=V$ or Player $\sigma$ fully wins $\mathcal{G}(V\setminus Attr_\sigma(c^{-1}(c'), \mathcal{A}))$, then Player $\sigma$ fully wins $\mathcal{G}$. 
        \item Otherwise, let $c'$ be a color in $C$  such that $Attr_\sigma(c^{-1}(c'), \mathcal{A})\ne V$ and  Player $\sigma$ doesn't fully win $\mathcal{G}(V\setminus Attr_\sigma(c^{-1}(c'), \mathcal{A}))$. Then we have $Win_\sigma(\mathcal{G})= Win_\sigma(\mathcal{G}(V\setminus X))$, where  $X=Attr_{\bar\sigma}(Win_{\bar\sigma}(\mathcal{G}(V\setminus Attr_\sigma(c^{-1}(c'), \mathcal{A}))), \mathcal{A})$. 
        \end{enumerate}
\end{lemma}
\begin{proof}
     For the first part of the lemma, assume that for all $c'\in c(V)$, $V=Attr_\sigma(c^{-1}(c'), \mathcal{A})$ or Player $\sigma$ fully wins the game $\mathcal{G}(V\setminus Attr_\sigma(c^{-1}(c'), \mathcal{A}))$. We construct the following winning strategy for Player $\sigma$ in $\mathcal{G}$. Let $c(V)=\{c_0, \dots, c_{k-1}\}$ and $i$ initially be 0.
    \begin{itemize}
        \item If the token is in $Attr_{\sigma}(c^{-1}(c_i),\mathcal{A})$, then Player $\sigma$ forces the token to a vertex in $c^{-1}(c_i)$ and once the token arrives at the vertex, sets $i=i+1 \mod k$. 
        \item Otherwise, Player $\sigma$ uses a winning strategy in  $\mathcal{G}(V\setminus Attr_{\sigma}(c^{-1}(c_i),\mathcal{A}))$. 
    \end{itemize}
    Consider any play consistent with the strategy described.     If there is an $i$ such that the token finally stays in $\mathcal{G}(V\setminus Attr_{\sigma}(c^{-1}(c_i),\mathcal{A}))$, then Player $\sigma$ wins the game. Otherwise, we have $c(\mathsf{Inf}(\rho)) = c(V)$. Since $c(V)\in \mathcal{F}_\sigma$,  Player $\sigma$ wins. This implies that Player $\sigma$ fully wins $\mathcal{G}$. 

For the second part,  let $c'\in C$ such that $Attr_\sigma(c^{-1}(c'), \mathcal{A})\ne V$ and Player $\sigma$ doesn't fully win $\mathcal{G}(V\setminus Attr_\sigma(c^{-1}(c'), \mathcal{A}))$. Let $V'=Win_{\bar\sigma}(\mathcal{G}(V\setminus Attr_\sigma(c^{-1}(c'), \mathcal{A})))$. Consider $X=Attr_{\bar\sigma}(V', \mathcal{A})$
as defined in the statement of the lemma. 
Note that  $\mathcal{A}(V')$ is a $\sigma$-trap in  $\mathcal{A}(V\setminus Attr_\sigma(c^{-1}(c'), \mathcal{A}))$; furthermore, $\mathcal{A}(V\setminus Attr_\sigma(c^{-1}(c'), \mathcal{A}))$ is a $\sigma$-trap in $\mathcal{A}$. This implies that  $\mathcal{A}(V')$ is a $\sigma$-trap in $\mathcal{A}$. Now we want to construct a winning strategy for Player $\bar\sigma$ in the arena $\mathcal A$ when the token is placed on $v\in X\cup Win_{\bar\sigma}(\mathcal{G}(V\setminus X))$. The winning strategy for Player $\bar\sigma$ in this case is the following:
    \begin{itemize}
        \item If $v\in X$, Player $\bar\sigma$ wins by forcing the token into $V'$ and following the winning strategy in $\sigma$-trap $\mathcal{A}(V')$.
        \item If $v\in Win_{\bar\sigma}(\mathcal{G}(V\setminus X))$, Player $\bar\sigma$ follows a winning strategy in $\mathcal{G}(Win_{\bar\sigma}(\mathcal{G}(V\setminus X)))$ until Player $\sigma$ moves the token into $X$.
    \end{itemize}
    Note that $\mathcal{A}(Win_\sigma(\mathcal{G}(V\setminus X)))$ is a $\bar\sigma$-trap in $\mathcal{A}(V\setminus X)$ and 
    $\mathcal{A}(V\setminus X)$ is a $\bar\sigma$-trap in $\mathcal{A}$. Hence, the set $\mathcal{A}(Win_\sigma(\mathcal{G}(V\setminus X)))$ is a $\bar\sigma$-trap in $\mathcal{A}$. Therefore, $Win_\sigma(\mathcal{G})= Win_\sigma(\mathcal{G}(V\setminus X))$.
\end{proof}

\noindent
As an immediate corollary we get the following lemma for Player $\sigma$.

\begin{lemma}\label{L: sigma fully win}
    Let $\mathcal G$ be a coloured Muller game and  let $\sigma\in \{0,1\}$ be such that $c(V)\in \mathcal{F}_\sigma$.  Player $\sigma$ fully wins $\mathcal{G}$ If and only if for all $c'\in c(V)$, $Attr_\sigma(c^{-1}(c'), \mathcal{A})=V$ or Player $\sigma$ fully wins $\mathcal{G}(V\setminus Attr_\sigma(c^{-1}(c'), \mathcal{A}))$. \qed 
\end{lemma}

Now we provide the next lemma that we call Trichotomy lemma. 
We will use this lemma in our dynamic programming based algorithms. 

\begin{lemma}[{\bf Trichotomy Lemma}] 
\label{L: 3 partitioning new}
    Let $\mathcal G$ be a coloured Muller game and  let $\sigma\in \{0,1\}$ be such that $c(V)\in \mathcal{F}_\sigma$.  Then we have the following two cases:
    \begin{enumerate}
    \item If for all $c'\in c(V)$, $Attr_\sigma(c^{-1}(c'), \mathcal{A})=V$ or Player $\sigma$ fully wins $\mathcal{G}(V\setminus Attr_\sigma(c^{-1}(c'), \mathcal{A}))$, then  Player $\sigma$ fully wins $\mathcal{G}$.
    \item Otherwise, if for all $v\in V$, $Attr_{\bar \sigma}(\{v\},\mathcal{A})=V$ or Player $\bar\sigma$ fully wins  $\mathcal{G}(V\setminus Attr_{\bar \sigma}(\{v\},\mathcal{A}))$, then Player $\bar\sigma$ fully wins $\mathcal{G}$.
    \item Otherwise, none of the players fully wins. 
    \end{enumerate}
\end{lemma}
\begin{proof}
By Lemma \ref{L: sigma fully win}, Part 1 is proved. For  the remaining parts of the lemma, we are under the assumption that Player $\sigma$ doesn't fully win $\mathcal{G}$.  For the second part, if Player $\bar\sigma$ fully wins $\mathcal{G}$, then for any $v\in V$, $Attr_{\bar \sigma}(\{v\},\mathcal{A})=V$ or Player $\bar\sigma$ fully wins the game in $\bar\sigma$-trap  $\mathcal{A}(V\setminus Attr_{\bar \sigma}(\{v\},\mathcal{A}))$. Otherwise, for all $v\in Win_{\bar\sigma}(\mathcal{G})$, $Attr_{\bar \sigma}(\{v\},\mathcal{A})\ne V$ and Player $\bar\sigma$ doesn't fully win  $\mathcal{G}(V\setminus Attr_{\bar \sigma}(\{v\},\mathcal{A}))$.
\end{proof}  

Note that Part 2 of the lemma assumes that Player $\sigma$ for which $c(V)\in \mathcal{F}_\sigma$ does not fully win the game. With this assumption, the second part characterizes the condition when Player $\bar \sigma$ fully wins the game; without this assumption, Part 2 does not hold true.

\section{Recursive algorithms for deciding regular games}\label{S: Algorithm 1}

Our goal is to provide recursive algorithms that solve regular games. To do so we utilise Lemma \ref{L: coloured muller characterization}. Naturally, we first start with a generic recursive algorithm that decides coloured Muller games, see Figure \ref{F: solve coloured Muller game 1 recursive}.  Lemma \ref{L: coloured muller characterization} guarantees correctness of the algorithm. Initially, the algorithm memorizes $\mathcal{G}$ globally. Then the function  $\text{SolveCMG}(V')$ is called. The algorithm returns  $(Win_0(\mathcal{G}(V')), Win_1(\mathcal{G}(V')))$. 

\begin{figure}[H]
    \centering 
    \scriptsize
    \begin{tabular}{l}
        \textbf{Global Storage:} A coloured Muller game $\mathcal{G}=(\mathcal{A},c,(\mathcal{F}_0, \mathcal{F}_1))$\\
        \textbf{Function:} $\text{SolveCMG}(V')$\\
        \textbf{Input}: A vertex set $V'$ with $\mathcal{A}(V')$ is an arena\\
        \textbf{Output}: ($Win_0(\mathcal{G}(V'))$, $Win_1(\mathcal{G}(V'))$) \\
        Let $\sigma\in \{0,1\}$ such that $c(V')\in \mathcal{F}_\sigma$;\\
        \textbf{for} $c'\in c(V')$ \textbf{do}\\
            \hspace*{4mm}$(W'_0,W'_1)\leftarrow \text{SolveCMG}(V'\setminus Attr_\sigma(c^{-1}(c'), \mathcal{A}(V')))$\\
            \hspace*{4mm}\textbf{if} $W'_\sigma\ne V'\setminus Attr_\sigma(c^{-1}(c'), \mathcal{A}(V'))$ \textbf{then}\\
                \hspace*{8mm}$X\leftarrow Attr_{\bar\sigma}(W'_{\bar\sigma}, \mathcal{A}(V'))$;\\ 
                \hspace*{8mm}$(W''_0,W''_1)\leftarrow \text{SolveCMG}(V'\setminus X)$;\\
                \hspace*{8mm}$W_\sigma\leftarrow W''_\sigma$,
                $W_{\bar\sigma}\leftarrow V'\setminus W_\sigma$;\\
                \hspace*{8mm}\textbf{return} ($W_0$, $W_1$)\\
            \hspace*{4mm}\textbf{end}\\
        \textbf{end}\\
        $W_\sigma\leftarrow V'$, $W_{\bar\sigma}\leftarrow \emptyset$;\\
        \textbf{return} ($W_0$, $W_1$)
    \end{tabular}
    \caption{The recursive algorithm for coloured Muller games}
    \label{F: solve coloured Muller game 1 recursive}
\end{figure}

A standard analysis of this algorithm produces running time $O(|C|^{|C|}\cdot |V|^{|V|})$, see  \cite{fijalkow2023games}. Our  
analysis below improves this by showing that the multiplicative factors $|C|^{|C|}$ and $|V|^{|V|}$ in this estimate can be replaced with $|C|!$ and $\binom{|V|}{|C|}$, respectively.  

\begin{lemma}\label{L: recursive times for SolveCMG}
    During the call of $\text{SolveCMG}(V)$, the function $\text{SolveCMG}$ is recursively called at most $|C|!\binom{|V|}{|C|}|V|$ times.
\end{lemma}
\begin{proof}
    If $|c(V')|=0$, then no $\text{SolveCMG}$ function is recursively called. Because $\mathcal{A}(V')$ is an arena, if $\text{SolveCMG}(V')$ is called then $|V'|\ne 1$.  If $|V'|=2$ then for all non-empty sets $V''\subseteq V'$ and $\sigma\in \{0,1\}$, $Attr_\sigma(V'',\mathcal{A}(V'))=V'$; hence, $\text{SolveCMG}$ is recursively called $|c(V')|$ times. If $|c(V')|=1$ then $\text{SolveCMG}$  is recursively called for $|c(V')|$ times. 
    
 Assume that $|V'|>2$, $|c(V')|>1$, and for all $V''$ with $|V''|< |V'|$, during the call of  $\text{SolveCMG}(V'')$, the function $\text{SolveCMG}$ is recursively called at most 
$|c(V'')|!\binom{|V''|}{|c(V'')|}|V''|$ times. For each $c'\in c(V')$, the set $V'\setminus Attr_\sigma(c^{-1}(c'), \mathcal{A}(V'))$ has at most $|V'|-1$ vertices and $|c(V')|-1$ colours. For $c'\in c(V')$ with 
$$
Win_\sigma(\mathcal{G}(V'\setminus Attr_\sigma(c^{-1}(c'), \mathcal{A}(V'))))\ne V'\setminus Attr_\sigma(c^{-1}(c'), \mathcal{A}(V')),
$$
we have $W'_{\bar\sigma} = Win_{\bar\sigma}(\mathcal{G}(V'\setminus Attr_\sigma(c^{-1}(c'), \mathcal{A}(V'))))$.  Let $X=Attr_{\bar\sigma}(W'_{\bar\sigma},\mathcal{A}(V'))$. Since $|W'_{\bar\sigma}|\geq 2$, the set $V'\setminus X$ contains at most $|V'|-2$ vertices and  
$|c(V')|$ colours. By hypothesis, during the call of $\text{SolveCMG}(V')$, the function $\text{SolveCMG}$ is recursively called at most 
\begin{center}{
$|c(V')|+1+|c(V')|(|c(V')|-1)!\binom{|V'|-1}{|c(V')|-1}(|V'|-1)+|c(V')|!\binom{|V'|-2}{|c(V')|}(|V'|-2)$}
\end{center}
times. This value is bounded from above by
$$
|c(V')|+1+|c(V')|!\binom{|V'|}{|c(V')|}(|V'|-1).
$$
Now there are 2 cases: 
\begin{enumerate}
    \item $|c(V')|=2$: Then
    \begin{align*}
        &|c(V')|!\binom{|V'|}{|c(V')|}|V'| - (|c(V')|+1+|c(V')|!\binom{|V'|}{|c(V')|}(|V'|-1))\\
        =&|c(V')|!\binom{|V'|}{|c(V')|}-|c(V')|-1\ge 2!\binom{3}{2}-3=3
    \end{align*}
    \item $|c(V')|>2$: Then
    \begin{align*}
        &|c(V')|!\binom{|V'|}{|c(V')|}|V'| - (|c(V')|+1+|c(V')|!\binom{|V'|}{|c(V')|}(|V'|-1))\\
        =&|c(V')|!\binom{|V'|}{|c(V')|}-|c(V')|-1\ge |c(V')|!-|c(V')|-1>0
    \end{align*}
\end{enumerate}
 Therefore, during the call of $\text{SolveCMG}(V')$, the function $\text{SolveCMG}$ is recursively called at most $|c(V')|!\binom{|V'|}{|c(V')|}|V'|$ times. By hypothesis, the proof is done.
\end{proof}

\begin{theorem}\label{Thm: SolveCMG time and space}
    There is an algorithm that,  
    given coloured Muller game $\mathcal{G}$ computes $Win_0(\mathcal{G})$ and $Win_1(\mathcal{G})$ in time $O(|C|!\binom{|V|}{|C|}|V||E|)$ and space $O(|\mathcal{G}|+|C||V|)$.
\end{theorem}
\begin{proof}
    Consider the algorithm in Figure \ref{F: solve coloured Muller game 1 recursive}. Apply $\text{SolveCMG}(V)$ to compute $Win_0(\mathcal{G})$ and $Win_1(\mathcal{G})$. The recursive depth of the algorithm is at most $|C|$ and $\mathcal{G}$ is memorized globally. In each iteration, only $O(|V|)$ space is applied to memorize the vertex set. Therefore, the algorithm takes $O(|\mathcal{G}|+|C||V|)$ space. By Lemma \ref{L: recursive times for SolveCMG}, the function $\text{SolveCMG}$ is recursively called for at most $|C|!\binom{|V|}{|C|}|V|$ times. We need to estimate the running time in two parts of the algorithm:
    \begin{itemize}
        \item {\em Part 1: The running time within the loop ``for $c'\in c(V')$ do''.}  In each enumeration of the color $c'$, there is a corresponding recursive call on $\text{SolveCMG}$. Every time when we have $W'_\sigma \ne V' \setminus Attr_\sigma(c^{-1}(c'), \mathcal{A}(V'))$, there is also a corresponding recursive call on $\text{SolveCMG}$. Since the function $\text{SolveCMG}$ is recursively called for at most $|C|!\binom{|V|}{|C|}|V|$ times, this part takes $O(|C|!\binom{|V|}{|C|}|V||E|)$ time.
        \item {\em Part 2: The running time outside the loop ``for $c'\in c(V')$ do''.} As $\text{SolveCMG}$ is recursively called for at most $|C|!\binom{|V|}{|C|}|V|$ times, the running time bound for this part of the algorithm is also  $O((|C|!\binom{|V|}{|C|}|V|+1)|V|)$.
    \end{itemize}
    Therefore, the algorithm takes $O(|C|!\binom{|V|}{|C|}|V||E|)$ time. Note that the correctness of the algorithm is provided by Lemma \ref{L: coloured muller characterization}.
\end{proof}

\subsection{Application to Rabin and Streett games}

Since Muller games are coloured Muller games in which each vertex has its own color, there is also a recursive algorithm for computing winning regions of Muller games. In this case, Lemma \ref{L: recursive times for SolveCMG} shows that the function $\text{SolveCMG}$ is recursively called at most $|V|!|V|$ times. Hence, Theorem \ref{Thm: SolveCMG time and space} implies the next lemma: 



\begin{lemma}\label{L: SolveMG time and space}
    There is a recurisve algorithm that,  
    given Muller game $\mathcal{G}$ computes $Win_0(\mathcal{G})$ and $Win_1(\mathcal{G})$ in time $O(|V|!|V||E|)$ and space $O(|\mathcal{G}|+|V|^2)$.
\end{lemma}

Through this lemma, by transforming Rabin conditions into Muller conditions, we can also provide a recursive algorithm for deciding Rabin games. The algorithm is presented in Figure \ref{F: solve Rabin games 1 recursive}.  

\begin{figure}[H]
    \centering 
    \scriptsize
    \begin{tabular}{l}
        \textbf{Global Storage:} A Rabin game $\mathcal{G}=(\mathcal{A},(U_1, V_1), \ldots, (U_k, V_k))$\\
        \textbf{Function:} $\text{SolveRG}(V')$\\
        \textbf{Input}: A vertex set $V'$ with $\mathcal{A}(V')$ is an arena\\
        \textbf{Output}: ($Win_0(\mathcal{G}(V'))$, $Win_1(\mathcal{G}(V'))$) \\
        If for all $i\in \{1, \ldots, k\}$ we have $V\cap U_i\neq \emptyset\implies V\cap V_i\neq \emptyset$ then $\sigma = 1$, otherwise $\sigma=0$.\\
        \textbf{for} $v\in V'$ \textbf{do}\\
            \hspace*{4mm}$(W'_0,W'_1)\leftarrow \text{SolveRG}(V'\setminus Attr_\sigma(\{v\}, \mathcal{A}(V')))$\\
            \hspace*{4mm}\textbf{if} $W'_\sigma\ne V'\setminus Attr_\sigma(\{v\}, \mathcal{A}(V'))$ \textbf{then}\\
                \hspace*{8mm}$X\leftarrow Attr_{\bar\sigma}(W'_{\bar\sigma}, \mathcal{A}(V'))$;\\ 
                \hspace*{8mm}$(W''_0,W''_1)\leftarrow \text{SolveRG}(V'\setminus X)$;\\
                \hspace*{8mm}$W_\sigma\leftarrow W''_\sigma$,
                $W_{\bar\sigma}\leftarrow V'\setminus W_\sigma$;\\
                \hspace*{8mm}\textbf{return} ($W_0$, $W_1$)\\
            \hspace*{4mm}\textbf{end}\\
        \textbf{end}\\
        $W_\sigma\leftarrow V'$, $W_{\bar\sigma}\leftarrow \emptyset$;\\
        \textbf{return} ($W_0$, $W_1$)
    \end{tabular}
    \caption{The recursive algorithm for Rabin games}
    \label{F: solve Rabin games 1 recursive}
\end{figure}

\begin{theorem}\label{Thm: Solve rabin and streett Recursive}
    We have the following:
    \begin{enumerate}
        \item There exists an algorithm that,  
    given Rabin or Streett game $\mathcal{G}$, computes $Win_0(\mathcal{G})$ and $Win_1(\mathcal{G})$ in time $O(|V|!|V|(|E|+k|V|))$ and space $O(|\mathcal{G}|+|V|^2)$.
    \item There exists an algorithm that,  
    given KL game $\mathcal{G}$ computes $Win_0(\mathcal{G})$ and $Win_1(\mathcal{G})$ in time $O(|V|!|V|(|E|+t|V|))$ and space $O(|\mathcal{G}|+|V|^2)$.
    \end{enumerate}
\end{theorem}
\begin{proof}
     Consider the algorithm above for Rabin games. We apply $\text{SolveRG}(V)$ to compute $Win_0(\mathcal{G})$ and $Win_1(\mathcal{G})$. Compared with the recursive algorithm of Muller games, the algorithm only changes the computing of $\sigma$. Therefore, the function $\text{SolveRG}$ is recursively called at most $|V|!|V|$ times. Also each computation of $\sigma$ takes $O(|k||V|)$ time. By Lemma \ref{L: SolveMG time and space}, the algorithm takes time $O(|V|!|V|(|E|+k|V|))$ and space $O(|\mathcal{G}|+|V|^2)$. For Streett games and KL games, similar arguments are applied. 
\end{proof}

\section{Dynamic programming algorithms for deciding regular games}\label{S: DP}

In this section, we provide dynamic programming algorithms for all regular games. First, in Section \ref{SS: Algorithm 1} we provide a dynamic version of the recursive algorithm in Figure \ref{F: solve coloured Muller game 1}. Then in 
Sections \ref{SS: Algorithm 2}--\ref{SS:Applications},  the next
set of all dynamic algorithms for solving  the regular games will utilize Lemma \ref{L: 3 partitioning new}.

Let $\mathcal{G}=(\mathcal{A},c,(\mathcal{F}_0, \mathcal{F}_1))$ be a coloured Muller game where $V=\{v_1,v_2,\ldots,v_{n}\}$. We need to code subsets of $V$ as binary strings. Therefore, we assign a $n$-bit binary number $i$ to each non-empty pseudo-arena $\mathcal{A}(S_i)$ in $\mathcal{G}$ so that $S_i=\{v_j\mid \text{the $j$th bit of $i$ is 1}\}$. We will use this notation for all our algorithms in this section. 

\subsection{Algorithm 1 for Coloured Muller Games}\label{SS: Algorithm 1}

Consider the algorithm in Figure \ref{F: solve coloured Muller game 1}. This is a dynamic programming version of the recursive algorithm in Figure \ref{F: solve coloured Muller game 1 recursive}. The algorithm, given a coloured Muller game $\mathcal{G}$ as input, returns the collections $Win_0(\mathcal{G})$ and $Win_1(\mathcal{G})$.
The correctness of the algorithm is guaranteed by Lemma \ref{L: coloured muller characterization}. Thus, we have the following theorem:

\begin{theorem}\label{Thm: Algorithm 1}
    There is an algorithm that solves coloured Muller game in time  $O(2^{|V|}|C||E|)$ and space $O(|\mathcal{G}|+2^{|V|}|V|)$.
\end{theorem}

\begin{proof}
    We use the Algorithm 1 in Figure \ref{F: solve coloured Muller game 1}. Note that we apply the binary trees to maintain $\mathcal F_\sigma$s, $W_0$ and $W_1$. For each $S_i$, the algorithm takes $O(|C||E|)$ time to compute $W_0(S_i)$ and $W_1(S_i)$. Therefore, this algorithm runs in $O(2^{|V|} |C||E|)$ time. Since $\mathcal F_\sigma$s, $W_0$ and $W_1$ are encoded by binary trees, the algorithm takes $O(|\mathcal{G}|+2^{|V|}|V|)$ space.   
\end{proof}

\begin{figure}[H]
    \centering 
    \scriptsize
    \begin{tabular}{l}
        \textbf{Input}: A coloured Muller game $\mathcal{G}=(\mathcal{A},c,(\mathcal{F}_0, \mathcal{F}_1))$\\
        \textbf{Output}: $Win_0(\mathcal{G})$, $Win_1(\mathcal{G})$ \\
        $W_0\leftarrow \emptyset$, $W_1\leftarrow \emptyset$;\\
        \textbf{for} $i=1$ to $2^n-1$ \textbf{do}\\
        \hspace*{4mm}$S_i\leftarrow \{v_j\mid \text{the $j$th bit of $i$ is 1}\}$;\\ 
        \hspace*{4mm}\textbf{if} $\mathcal{A}(S_i)$ is not an arena \textbf{then}\\
        \hspace*{8mm}\textbf{break};\\
        \hspace*{4mm}\textbf{end}\\
        \hspace*{4mm}Let $\sigma\in \{0,1\}$ such that $c(S_i)\in \mathcal{F}_\sigma$;\\
        \hspace*{4mm}$is\_win=$\text{true};\\
        \hspace*{4mm}\textbf{for} $c'\in c(S_i)$ \textbf{do}\\
            \hspace*{8mm}\textbf{if} $Attr_\sigma(c^{-1}(c'), \mathcal{A}(S_i))\ne S_i$ and\\
            \hspace*{10mm}$W_\sigma(S_i\setminus Attr_\sigma(c^{-1}(c'), \mathcal{A}(S_i)))\ne S_i\setminus Attr_\sigma(c^{-1}(c'), \mathcal{A}(S_i))$ \textbf{then}\\
                \hspace*{12mm}$is\_win=$\text{false};\\
                \hspace*{12mm}$X\leftarrow Attr_{\bar\sigma}(W_{\bar\sigma}(\mathcal{G}(S_i\setminus Attr_\sigma(c^{-1}(c'), \mathcal{A}(S_i)))), \mathcal{A}(S_i))$;\\ 
                \hspace*{12mm}$W_\sigma(S_i)\leftarrow W_\sigma(S_i\setminus X)$, $W_{\bar\sigma}(S_i)\leftarrow S_i\setminus W_\sigma(S_i)$;\\
                \hspace*{12mm}\textbf{break};\\
            \hspace*{8mm}\textbf{end}\\
        \hspace*{4mm}\textbf{end}\\
        \hspace*{4mm}\textbf{if} $is\_win=$\text{true} \textbf{then}\\
            \hspace*{8mm}$W_{\sigma}(S_i)\leftarrow S_i$, $W_{\bar\sigma}(S_i)\leftarrow \emptyset$;\\
        \hspace*{4mm}\textbf{end}\\
        \textbf{end}\\
        \textbf{return} $W_0(V)$ and $W_1(V)$
    \end{tabular}
    \caption{Algorithm 1 for coloured Muller games}
    \label{F: solve coloured Muller game 1}
\end{figure}

\subsection{Algorithm 2 for Coloured Muller Games}\label{SS: Algorithm 2}

In this section, we utilise the concept of full win for the players, see Definition \ref{Dfn:FullWin}. The new dynamic algorithm, Algorithm 2, is presented in Figure \ref{F:Partition coloured Muller game 2}. The algorithm takes coloured Muller game $\mathcal{G}=(\mathcal{A},c,(\mathcal{F}_0, \mathcal{F}_1))$ as input.  Lemma \ref{L: 3 partitioning new} guarantees correctness of the algorithm. During the running process, this dynamic algorithm  partitions all subgames $\mathcal{G}(S_i)$ into the following three collections of subsets of $V$:
\begin{itemize}
    \item $P_0=\{S_i\mid i\in [1,2^n-1] \text{ and Player 0 fully wins }\mathcal{G}(S_i)\}$, 
    \item $P_1=\{S_i\mid i\in [1,2^n-1]\text{ and Player 1 fully wins }\mathcal{G}(S_i)\}$, and
    \item  $Q=\{S_i\mid i\in [1,2^n-1]\text{ and no player fully wins }\mathcal{G}(S_i)\}$.
\end{itemize}    

\noindent
Now we provide analysis of Algorithm 2 presented in Figure \ref{F:Partition coloured Muller game 2}.

\begin{figure}[H]
    \centering 
    \scriptsize
    \begin{tabular}{l}
        \textbf{Input}: A coloured Muller game $\mathcal{G}=(\mathcal{A},c,(\mathcal{F}_0, \mathcal{F}_1))$\\
        \textbf{Output}: The partitioned sets $P_0$, $P_1$ and $Q$.\\
        $P_0\leftarrow \emptyset$, $P_1\leftarrow \emptyset$, $Q\leftarrow \emptyset$;\\
        \textbf{for} $i=1$ to $2^n-1$ \textbf{do}\\
        \hspace*{4mm}$S_i\leftarrow \{v_j\mid \text{the $j$th bit of $i$ is 1}\}$;\\ 
        \hspace*{4mm}\textbf{if} $\mathcal{A}(S_i)$ is not an arena \textbf{then}\\
        \hspace*{8mm}\textbf{break};\\
        \hspace*{4mm}\textbf{end}\\
        \hspace*{4mm}Let $\sigma\in \{0,1\}$ such that $c(S_i)\in \mathcal{F}_\sigma$;\\
        \hspace*{4mm}$AllAttr_0=$\text{true}, $AllAttr_1=$\text{true};\\
        \hspace*{4mm}\textbf{for} $c'\in c(S_i)$ \textbf{do}\\
            \hspace*{8mm}\textbf{if} $Attr_\sigma(c^{-1}(c'), \mathcal{A}(S_i))\ne S_i$ and $S_i\setminus Attr_\sigma(c^{-1}(c'), \mathcal{A}(S_i)) \notin P_\sigma$ \textbf{then}\\
                \hspace*{12mm}$AllAttr_\sigma=$\text{false};\\
                \hspace*{12mm}\textbf{break}\\
            \hspace*{8mm}\textbf{end}\\
        \hspace*{4mm}\textbf{end}\\
        \hspace*{4mm}\textbf{if} $AllAttr_\sigma=$\text{true} \textbf{then}\\
            \hspace*{8mm}$P_\sigma\leftarrow P_\sigma \cup \{S_i\}$;\\
        \hspace*{4mm}\textbf{else}\\
            \hspace*{8mm}\textbf{for} $v\in S_i$ \textbf{do}\\
            \hspace*{12mm}\textbf{if} $Attr_{\bar\sigma}(\{v\}, \mathcal{A}(S_i))\ne S_i$ and $S_i\setminus Attr_{\bar\sigma}(\{v\}, \mathcal{A}(S_i)) \notin P_{\bar\sigma}$ \textbf{then}\\
                \hspace*{16mm}$AllAttr_{\bar\sigma}=$\text{false};\\
                \hspace*{16mm}\textbf{break}\\
            \hspace*{12mm}\textbf{end}\\
            \hspace*{8mm}\textbf{if} $AllAttr_{\bar\sigma}=\text{true}$ \textbf{then}\\
                \hspace*{12mm}$P_{\bar\sigma}\leftarrow P_{\bar\sigma} \cup \{S_i\}$;\\
            \hspace*{8mm}\textbf{else}\\
                \hspace*{12mm}$Q\leftarrow Q \cup \{S_i\}$;\\
            \hspace*{8mm}\textbf{end}\\
        \hspace*{4mm}\textbf{end}\\
        \textbf{end}\\
        \textbf{return} $P_0$, $P_1$ and  $Q$
    \end{tabular}
    \caption{Algorithm 2 for partitioning subgames of a coloured Muller game}
    \label{F:Partition coloured Muller game 2}
\end{figure}

\begin{lemma}\label{L: P0 P1 and Q implement}
    Algorithm 2 computes $P_0$, $P_1$ and $Q$ for a coloured Muller game in $O(2^{|V|} |V||E|)$ time  and  $O(|\mathcal{G}|+2^{|V|})$ space.
\end{lemma}

\begin{proof}
        We use the Algorithm 2 in Figure \ref{F:Partition coloured Muller game 2}. Note that we apply the binary trees to maintain $\mathcal F_\sigma$s, $P_0$, $P_1$ and $Q$. For each $S_i$, the algorithm takes $O(|V||E|)$ time to determine the set to add $S_i$. Therefore, this algorithm runs in $O(2^{|V|} |V||E|)$ time. Since $P_0$, $P_1$ and $Q$ are encoded by binary trees, the algorithm takes $O(|\mathcal{G}|+2^{|V|})$ space.   
\end{proof}

\begin{lemma}\label{L: 1-trap union}
    Let  $\mathcal{A}(X)$ and $\mathcal{A}(Y)$ be 1-traps. 
    If Player 0 fully wins $\mathcal{G}(X)$ and $\mathcal{G}(Y)$ then Player 0 fully wins $\mathcal{G}(X\cup Y)$. 
\end{lemma}
\begin{proof}
    We construct a winning strategy for Player 0 in $\mathcal{G}(X\cup Y)$ as follows. \  If the token is in $Attr_0(X, \mathcal{A}(X\cup Y))$, Player 0 forces the token into $X$ and once the token arrives at $X$, Player 0 follows the winning strategy in $\mathcal{G}(X)$. \ Otherwise, Player 0 follows the winning strategy in $\mathcal{G}(Y)$.
\end{proof}

\begin{lemma}\label{L: P_0 to winning region}
    If for all $S_i\in P_0$, the arena $\mathcal{A}(S_i)$ isn't 1-trap in $\mathcal{G}$, then $Win_0(\mathcal{G})=\emptyset$ and $Win_1(\mathcal{G})=V$. Otherwise, let $\mathcal{A}(S_{max})$ be the maximal 1-trap in $\mathcal{G}$ so that $S_{max}\in P_0$. Then $Win_0(\mathcal{G})=S_{max}$ and $Win_1(\mathcal{G})=V\setminus S_{max}$.
\end{lemma}
\begin{proof}
    For the first part of the lemma, assume  that $Win_0(\mathcal{G})\ne\emptyset$.  Now note that $\mathcal{A}(Win_0(\mathcal{G}))$ is 1-trap such that
     Player 0 fully wins $\mathcal G (Win_0(\mathcal{G}))$. This contradicts with the assumption of the first part. 
For the second part, consider all 1-traps $\mathcal{A}(X)$ with $X\in P_0$. \  Player 0 fully wins the games $\mathcal G(X)$ in each of these 1-traps by definition of $P_0$. By Lemma \ref{L: 1-trap union}, Player 0 fully wins the union of these 1-traps. Clearly, this union is $S_{max} \in P_0$. Consider $V\setminus S_{max}$. This set determines a $0$-trap. Suppose Player 1 does not win 
$\mathcal G(V\setminus S_{max})$ fully. Then there exists 
a $1$-trap $\mathcal{A}(Y)$ in game $\mathcal G(V\setminus S_{max})$ such that Player 0 fully wins $\mathcal G(Y)$. For every Player 1 position in $y\in Y$ and outgoing edge $(y,x)$ we have either $x\in Y$ or $x\in S_{max}$. This implies $\mathcal{A}(S_{max}\cup Y)$ is 1-trap  such that Player 0 fully wins $\mathcal G(S_{max}\cup Y)$. So, $S_{max}\cup Y$ must be in $P_0$. This contradicts with the choice of $S_{max}$. 
\end{proof}

\noindent

By Lemmas  \ref{L: P0 P1 and Q implement} and \ref{L: P_0 to winning region}, we have proved the following theorem.

\begin{theorem}\label{Thm: Algorithm 2}
There exists an algorithm that decides the coloured Muller games  $\mathcal{G}$ in time $O(2^{|V|} |V||E|)$  and space $O(|\mathcal{G}|+2^{|V|})$.
\end{theorem}

\subsection{Applications to Muller and McNaughton games} \label{SS:AppMM}

It is not too hard to see that for Muller games and McNaughton games, we can easily recast the algorithms presented in Sections \ref{SS: Algorithm 1} and \ref{SS: Algorithm 2}. Indeed, the transformation of Muller games to coloured Muller games is obvious. Hence, by applying  Theorem  \ref{Thm: Algorithm 2}  to Muller games we get the following result:

\begin{theorem} \label{Thm:Muller-DP}
There exists an algorithm that decides Muller game $\mathcal{G}$ in time $O(2^{|V|}|V||E|)$ and space $O(|\mathcal{G}|+2^{|V|})$.  \qed 
\end{theorem}

The transformation of McNaughton games into coloured Muller games is also easy. Each position $v$ in $W$ gets its own color, and all positions outside of $W$ get the same new colour.  Hence, we  can apply both Theorems \ref{Thm: Algorithm 1} and \ref{Thm: Algorithm 2} to McNaughton games:

\begin{theorem}\label{Thm: solve Mcnaughton game DP}
Each of the following is true:
\begin{enumerate}
\item  There exists an algorithm that decides McNaughton games $\mathcal{G}$ in $O(2^{|V|}|W||E|)$ time  and  $O(|\mathcal{G}|+2^{|V|}|V|)$ space.
\item There exists an algorithm that decides McNaughton games $\mathcal{G}$  in $O(2^{|V|} |V||E|)$ time  and $O(|\mathcal{G}|+2^{|V|})$ space. \qed 
\end{enumerate}
\end{theorem}

\subsection{Enumeration Lemma}



This is an auxiliary section that will provide us with an enumeration technique. This technique will then be used in designing an algorithm to decide Rabin and Streett games by transforming these games into Muller games in a more  efficient manner. 

Let $n$ be a natural number and $\mathcal{S}=\{b_1,\ldots, b_t\}$ be a set of $n$-bit binary integers, where $n$ is the size of the vertex set $V=\{v_1, \ldots, v_n\}$ of the arena. Each $b_i$ represents the characteristic function of the set $V_i\subseteq V$: $b_i(v)=1$ iff $v\in V_i$.  We want to efficiently enumerate the collection $2^{V_1} \cup \ldots \cup 2^{V_t}$. Note that
    $$
    2^{V_1} \cup \ldots \cup 2^{V_t}=\{x\in [0,2^n)\mid \exists 
    b \in \mathcal{S} (x \ \& \  b=x) \},
    $$
    where $x$ is the binary integer of length at most $n$,  and the operation $\&$ is the bitwise {\em and} operation. Later we will use our enumeration of the collection 
    $$\mathcal X=\{x\in [0,2^n)\mid \exists 
    b \in \mathcal{S} (x \ \& \ b=x) \}
    $$ 
    to transform the KL condition into Muller condition. 

    Note that the brute-force algorithm that enumerates the collection $\mathcal X=2^{S_1} \cup \ldots \cup 2^{S_t}$ runs in time $O(2^n \cdot t)$. In our enumeration we want to remove the dependence on $t$ as $t$ can be exponential on $n$. This is done in the next lemma: 

\begin{lemma}[{\bf Enumeration Lemma}]\label{L: enmerate time complexity}
    Given the set $\mathcal{S}=\{b_1,\ldots, b_t\}$ of $n$-bit binary integers, we can enumerate the collection 
    $\mathcal X=\{x\in [0,2^n)\mid \exists 
    b \in \mathcal{S} (x \ \&  \ b=x) \}$  in time $O(2^{n} n)$ and space $O(2^{n})$.
\end{lemma}

\begin{proof}
    We apply the function $\text{Enumerate}(\mathcal{S}, n)$ shown in Figure \ref{F:Enumerate}. Also we apply the binary trees to maintain sets of $n$-bit binary integers such as $\mathcal{S}$, $\mathcal{X}$.

    \begin{figure}[H]
    \centering 
    \scriptsize
    \begin{tabular}{l}
        \textbf{Function}: $\text{Enumerate}(\mathcal{S}, n)$.\\
        \textbf{Input}: $\mathcal{S}$ and $n$ where $\mathcal{S}$ is a set of $n$-bit binary integers.\\
        \textbf{Output}: $\mathcal{X}=\{x\in [0,2^n)\mid \exists b\in \mathcal{S} (x \ \& \ b=x) \}$.\\
        \textbf{if} $\mathcal{S}=\emptyset$ \textbf{then}\\
            \hspace*{4mm}\textbf{return} $\emptyset$\\
        \textbf{end}\\
        \textbf{if} $n=0$ \textbf{then}\\
            \hspace*{4mm}\textbf{return} $\{0\}$\\
        \textbf{end}\\
        $\mathcal{S}'_0\leftarrow \emptyset$, $\mathcal{S}'_1\leftarrow \emptyset$\\
        \textbf{for} $b\in \mathcal{S}$ \textbf{do}\\
            \hspace*{4mm}\textbf{if} $b \mod 2 =0$ \textbf{then}\\
                \hspace*{8mm}$\mathcal{S}'_0\leftarrow \mathcal{S}'_0\cup \{\frac{b}{2}\}$\\
            \hspace*{4mm}\textbf{else}\\
                \hspace*{8mm}$\mathcal{S}'_0\leftarrow \mathcal{S}'_0\cup \{\frac{b-1}{2}\}$\\
                \hspace*{8mm}$\mathcal{S}'_1\leftarrow \mathcal{S}'_1\cup \{\frac{b-1}{2}\}$\\
            \hspace*{4mm}\textbf{end}\\
        \textbf{end}\\
        $\mathcal{X}'_0\leftarrow \text{Enumerate}(\mathcal{S}'_0, n-1)$, $\mathcal{X}'_1\leftarrow \text{Enumerate}(\mathcal{S}'_1, n-1)$, $\mathcal{X}\leftarrow \emptyset$\\
        \textbf{for} $x'\in \mathcal{X}'_0$ \textbf{do}\\
            \hspace*{4mm}$\mathcal{X}\leftarrow \mathcal{X}\cup \{2x'\}$ \\
        \textbf{end}\\
        \textbf{for} $x'\in \mathcal{X}'_1$ \textbf{do}\\
            \hspace*{4mm}$\mathcal{X}\leftarrow \mathcal{X}\cup \{2x'+1\}$ \\
        \textbf{end}\\
        \textbf{return} $\mathcal{X}$
    \end{tabular}
    \caption{Algorithm for $\text{Enumerate}(\mathcal{S}, n)$}
    \label{F:Enumerate}
\end{figure}

    Let $(\mathcal{S}_1, n_1),\ldots, (\mathcal{S}_k, n_k)$ be the sequence of all inputs recursively computed during the execution of  $\text{Enumerate}(\mathcal{S}, n)$ where $n_1,\ldots ,n_k$ are in non-decreasing order. In the following, we want to show that for each $(\mathcal{S}_i, n_i)$,  $\text{Enumerate}(\mathcal{S}_i,n_i)=\{x\in [0,2^{n_{i}})\mid \exists_{b\in \mathcal{S}_{i}} (x\ \&\ b=x) \}$. Not hard to see that $n_1=0$ or $\mathcal{S}_1=\emptyset$ as otherwise there is a recursion on computing an input with smaller $n$. If $\mathcal{S}_1=\emptyset$ then $\text{Enumerate}(\mathcal{S}_1, n_1)=\emptyset = \{x\in [0,2^{n_1})\mid \exists_{b\in \mathcal{S}_1} (x\ \&\ b=x) \}$, otherwise $n_1=0$ and $\text{Enumerate}(\mathcal{S}_1, n_1)=\{0\} = \{x\in [0,2^{n_1})\mid \exists_{b\in \mathcal{S}_1} (x\ \&\ b=x) \}$. Then we consider $\text{Enumerate}(\mathcal{S}_i, n_i)$ with $i>1$ and assume for all $i'<i$, $\text{Enumerate}(\mathcal{S}_{i'}, n_{i'})= \{x\in [0,2^{n_{i'}})\mid \exists_{b\in \mathcal{S}_{i'}} (x\ \&\ b=x) \}$. 
    \begin{itemize}
        \item $\mathcal{S}_i=\emptyset$ or $n_i=0$: If $\mathcal{S}_i=\emptyset$ then $\text{Enumerate}(\mathcal{S}_i, n_i)=\emptyset = \{x\in [0,2^{n_i})\mid \exists_{b\in \mathcal{S}_i} (x\ \&\ b=x) \}$, otherwise $n_i=0$ and $\text{Enumerate}(\mathcal{S}_i, n_i)=\{0\} = \{x\in [0,2^{n_i})\mid \exists_{b\in \mathcal{S}_i} (x\ \&\ b=x) \}$.
        \item Otherwise: Let $\mathcal{S}'_{0}=\{\lfloor \frac{b}{2}\rfloor \mid b\in \mathcal{S}_i\}$ and $\mathcal{S}'_{1}=\{\frac{b-1}{2}\mid b\in \mathcal{S}_i\text{ and } b\mod 2=1\}$. Then by hypothesis, $\text{Enumerate}(\mathcal{S}'_0, n_i-1)=\{x\in [0, 2^{n_i-1})\mid \exists_{b\in \mathcal{S}'_0} (x\ \&\  b$ $=x)\}$ and $\text{Enumerate}(\mathcal{S}'_1, n_i-1)=\{x\in [0, 2^{n_i-1})\mid \exists_{b\in \mathcal{S}'_1} (x\ \&\ b=x)\}$. Therefore $\text{Enumerate}(\mathcal{S}_i,n_i)$ $=\{2x \mid x\in [0, 2^{n_i-1})\text{ and }\exists_{b\in \mathcal{S}'_0} (x\ \&\ b=x)\}\cup \{2x+1\mid x\in [0, 2^{n_i-1})\text{ and }\exists_{b\in \mathcal{S}'_1} (x\ \&\ b=x)\}=\{x\in [0, 2^{n_i}) \mid x\mod 2=0\text{ and }\exists_{b\in \mathcal{S}_i} (x\ \& \ b=x)\}\cup \{x\in [0, 2^{n_i})\mid x\mod 2=1\text{ and }\exists_{b\in \mathcal{S}_i} (x\ \&\ b=x)\}= \{x\in [0,2^{n_{i}})\mid \exists_{b\in \mathcal{S}_{i}} (x\ \&\ b=x) \}$.
    \end{itemize}
    Therefore,  $ \text{Enumerate}(\mathcal{S}_i,n_i)=\{x\in [0,2^{n_{i}})\mid \exists_{b\in \mathcal{S}_{i}} (x\ \&\ b=x) \}$. By hypothesis, we show that $\text{Enumerate}(\mathcal{S}, n)=\{x\in [0,2^{n})\mid \exists_{b\in \mathcal{S}}(x\ \&\ b=x)\}$.

    Then we measure the complexity of running $\text{Enumerate}(\mathcal{S},n)$. We partition $\{\mathcal{S}_1,n_1\}\ldots \{\mathcal{S}_k,n_k\}$ into $P(n')=\{\mathcal{S}_i\mid i\in [1,k] \text{ and } n_i=n'\}$ for $n'\in [0,n]$. Then for each $(\mathcal{S}_i, n_i)$ with $n_i>0$, it calls $\text{Enumerate}(\mathcal{S'}, n')$ with $n'=n_i-1$ at most 2 times. Therefore, for each $n'\in [0,n]$, $|P(n')|\le 2^{n-n'}$. For each $(\mathcal{S}_i, n_i)$, $\mathcal{S}_i$ is maintained by a binary tree rooted by $r_i$ where the subtrees of $r_i$ are $\mathcal{S}_{i,0}=\{\frac{b}{2}\mid b\in \mathcal{S}_i \text{ and }  b\mod 2 = 0\}$ and $\mathcal{S}_{i,1}=\{\frac{b-1}{2}\mid b\in \mathcal{S}_i \text{ and }  b\mod 2 = 1\}$. Then for each $(\mathcal{S}_i, n_i)$, $\mathcal{S}'_0=\mathcal{S}_{i,0}\cup \mathcal{S}_{i,1}$ and $\mathcal{S}'_1= \mathcal{S}_{i,1}$ are computed in $O(2^{n_i})$ time through the union of binary trees. Similarly, for each $(\mathcal{S}_i, n_i)$, $\mathcal{X}$ is computed in $O(2^{n_i})$ time. Hence, the time complexity of each iteration $\text{Enumerate}(\mathcal{S}_i, n_i)$ is bounded by $O(2^{n_i})$.  Since $\sum_{n'\in [0,n]}\sum_{\mathcal{S}'\in P(n')}2^{n'}= \sum_{n'\in [0,n]}2^{n-n'}\cdot 2^{n'}= 2^n\cdot n$, the time complexity of the algorithm is bounded by $O(2^n\cdot n)$. Since the maximum recursion depth of the algorithm is $n$ and for the recursion at level $i$, $O(2^i)$ space is applied, the algorithm takes $O(2^n)$ space.
\end{proof}

\subsection{Applications to Rabin and Streett games}\label{SS:Applications}

We can naturally transform Rabin games, Streett games, and KL games into Muller games, and then apply  our dynamic algorithms from Section \ref{SS:AppMM} to thus obtained Muller games.  These transformations are the following:
\begin{itemize}

\item For Rabin games and $X\subseteq V$, if for $i\in \{1, \ldots, k\}$ we have $X\cap U_i\neq \emptyset\implies X\cap V_i\neq \emptyset$ then $X\in \mathcal{F}_1$, otherwise $X\in \mathcal{F}_0$. 

\item For Streett games and $X\subseteq V$, if there is an $i\in\{1, \ldots, k\} $ such that $X\cap U_i \neq \emptyset$ and $X\cap V_i=\emptyset$, then $X\in \mathcal{F}_1$, otherwise $X\in \mathcal{F}_0$. 

\item For KL games and $X\subseteq V$, if for $i\in \{1,\ldots,t\}$ we have $u_i\in X\implies X\not\subseteq S_i$ then $X\in \mathcal{F}_1$, otherwise $X\in \mathcal{F}_0$.
\end{itemize}

In these transformations one needs to be careful with the parameters $k$ and $t$ for Rabin and Streett games and KL games, respectively. They add additional running time costs,  especially $k$ and $t$ can have exponential values in $|V|$. For instance, the direct translation of Rabin games to  Muller games requires, for each pair $(U_i,V_i)$ in the Rabin winning condition, to build the collection of sets $X$ such that $X\cap U_i\neq \emptyset$ and $X\cap V_i=\emptyset$. The collection of all these sets $X$ form the Muller condition set $(\mathcal{F}_0, \mathcal{F}_1)$. As the index $k$ is $O(2^{2|V|})$, the direct transformation above is expensive. Our goal now is to carefully analyse the transformations of Rabin games to Muller games.

We start with transforming KL games to Muller games. 
Let $\mathcal{G}=(\mathcal{A},(u_1, S_1),$ $\ldots , (u_t, S_t))$ be a KL game. Define Muller game $\mathcal G'$, where $(\mathcal{F}_0, \mathcal{F}_1)$ are given as follows:
$$X\in \mathcal{F}_0\text{ if for some pair }(u_i, S_i) \text{ we have } u_i\in X\text{ and } X\subseteq S_i, \text{ otherwise } X\in \mathcal{F}_1$$

\begin{lemma}\label{L: transformation from KL to Muller}
    The transformation from KL games $\mathcal{G}$ to Muller games $\mathcal{G}'$ takes $O(2^{|V|}|V|^2)$ time and $O(|\mathcal{G}|+2^{|V|})$ space.
\end{lemma}

\begin{proof}
We apply the binary encoding so that for $i\in[1,t]$, $u_i\in [0,n)$ and $S_i\in [0,2^n)$. In the following, we apply binary trees to maintain sets of binary integers. We transform $\mathcal{G}$ into Muller game $\mathcal{G}'=(\mathcal{A}, (\mathcal{F}_0, \mathcal{F}_1))$ where we also apply the binary encoding so that $\mathcal{F}_0=\{X\in [0, 2^n)\mid \text{there exists an } i\in [1,t] \text{ so that the } u_i\text{-th bit of } X$ $\text{is 1 and } X \ \&\ S_i=X\}$ and $\mathcal{F}_1=\{0,1,\ldots, 2^n-1\}\setminus \mathcal{F}_0$. Then let $\mathcal{S}_{i}=\{S_j\mid j\in [1,t]\text{ and } u_j=i\}$ for $i\in [0,n)$. By Lemma \ref{L: enmerate time complexity}, for each $i\in [0,n)$,  we compute $\{X\in [0,2^n)\mid \exists_{S\in \mathcal{S}_i} X\ \&\  S=X\}$ in time $O(2^n\cdot n)$ and space $O(2^n)$, and then compute  $\{X\in [0,2^n)\mid \text{the }i\text{-th bit of }X\text{ is 1} \text{ and }\exists_{S\in \mathcal{S}_i} X\ \&\ S=X\}$ in time $O(2^n)$ by traversing the binary trees, checking and deleting subtrees at depth $i$.  Since $\bigcup_{i\in [0,n)}\{X\in [0,2^n)\mid \text{the }i\text{-th bit of }X\text{ is 1} \text{ and }\exists_{S\in \mathcal{S}_i} X\ \&\ S=X\}=\mathcal{F}_0$, 
    we reuse $O(2^n)$ space for each $i\in [0,n)$ and use another $O(2^{n})$ space to record the prefix union results. Since $\mathcal{F}_1$ is computed from $\mathcal{F}_0$ in time $O(2^{|V|})$ by computing the complement of the tree, this is a transformation from KL games to Muller games and the transformation takes $O(2^{|V|} |V|^2)$ time and $O(|\mathcal{G}|+2^{|V|})$ space.
\end{proof}

As an immediate corollary we get the following complexity-theoretic result for KL games. 
\begin{theorem} \label{Thm: decide KL 1}
 There exists an algorithm that, given a KL game $\mathcal G$, decides $\mathcal G$ in  
 $O(2^{|V|}|V||E|)$ time and  $O(|\mathcal{G}|+2^{|V|})$ space. \qed 
\end{theorem}

Now we transform Rabin games $\mathcal G$ to Muller games. As we mentioned above, the direct translation to Muller games is costly.  
Our goal is to avoid this cost through KL games.  The following lemma is easy:

\begin{lemma}\label{L: Ui Vi to Yi Zi}
     Let $X\subseteq V$ and let $(U_i, V_i)$ be a winning pair in Rabin game $\mathcal G$.
     Set $Y_i=U_i\setminus V_i$ and $Z_i=V\setminus V_i$. Then $X\cap U_i\ne\emptyset$ and $X\cap V_i=\emptyset$ if and only if $X\cap Y_i\ne\emptyset$ and $X\subseteq Z_i$.
\end{lemma}

Thus, we can replace the winning condition $(U_1, V_1), \ldots (U_k, V_k)$ in  Rabin games  to the equivalent winning condition $(Y_1, Z_1), \ldots, (Y_k, Z_k)$. We still have Rabin winning condition but we use this 
new winning condition $(Y_1, Z_1), \ldots, (Y_k, Z_k)$ to build the desired KL game:  

\begin{lemma} \label{L:R-to-KL}
    The transformation from Rabin games $\mathcal{G}$ to KL games takes time $O(k|V|^2)$  and space $O(|\mathcal{G}|+2^{|V|}|V|)$.
\end{lemma}
\begin{proof}
    Enumerate all pairs $(U_i, V_i)$, compute $Y_i=U_i\setminus V_i$, $Z_i=V\setminus V_i$ and add all pairs $(u_j, S_j)$ with $u_j\in Y_i$ and $S_j=Z_i$ into KL conditions. By applying binary trees, the transformation takes $O(k|V|^2)$ time and $O(|\mathcal{G}|+2^{|V|}|V|)$ space. This preserves the winning sets $W_0$ and $W_1$. 
\end{proof}

Thus, the transformed KL games can be viewed as a compressed version of Rabin games. 

\begin{corollary}
    The transformation from Rabin games $\mathcal{G}$ to Muller games $\mathcal{G}'$ takes $O((k+2^{|V|})|V|^2)$ time  and $O(|\mathcal{G}|+2^{|V|}|V|)$ space. \qed
\end{corollary}

Note that deciding Rabin games is equivalent to deciding Streett games. Thus, combining the arguments above, we get the following complexity-theoretic result:
\begin{theorem}\label{Thm: Solve rabin and streett DP}
 There exist algorithms that decide Rabin and Streett games $\mathcal{G}$ in $O((k|V|+2^{|V|}|E|)|V|)$ time  and $O(|\mathcal{G}|+2^{|V|}|V|)$ space. \qed
\end{theorem}



\section{Conclusion}

The algorithms presented in this work give rise to numerous questions that warrant further exploration. For instance, we know that explicitly given Muller games can be decided in polynomial time. Yet, we do not know if there are polynomial time algorithms that decide explicitly given  McNaughton games. Another intriguing line of research is to investigate if there are exponential time algorithms that decide coloured Muller games when the parameter $|C|$ ranges in the interval $[\sqrt{|V|}, |V|/a]$, where $a>1$. It could also be very interesting to replace the factor $2^{|V|}$ with $2^{|W|}$ in the running time that decides McNaughton games.  If so, this implies that the ETH  is not applicable to McNaughton games as opposed to coloured Muller games (and Rabin games).  These all may uncover new insights and lead to even more efficient algorithms.

\bibliographystyle{siamplain}
\bibliography{bibfile}
\end{document}

%% file: ex_shared.tex
\usepackage{lipsum}
\usepackage{amsfonts}
\usepackage{graphicx}
\usepackage{epstopdf}
\usepackage{enumitem}
\usepackage{float}
\usepackage{color}
\usepackage{algorithm}
\usepackage{algorithmic}
\usepackage{diagbox}
\usepackage{multirow}
\usepackage{subfigure}
\usepackage{amsmath}
\ifpdf
  \DeclareGraphicsExtensions{.eps,.pdf,.png,.jpg}
\else
  \DeclareGraphicsExtensions{.eps}
\fi

\newsiamthm{agreement}{agreement}
\newsiamthm{claim}{Claim}
\newsiamthm{example}{Example}

\title{Deciding regular games:\\ 
a playground for exponential time algorithms}

\author{Zihui Liang\thanks{University of Electronic Science and Technology of China
(\email{zihuiliang.tcs@gmail.com}).}
\and Bakh Khoussainov\thanks{University of Electronic Science and Technology of China
 (\email{bmk@uestc.edu.cn}).}
\and Mingyu Xiao\thanks{University of Electronic Sciences and Technology of China
(\email{myxiao@uestc.edu.cn}).} }

\usepackage{amsopn}


%% file: main.bbl
\begin{thebibliography}{10}

\bibitem{bjorklund2003fixed}
{\sc H.~Bj{\"o}rklund, S.~Sandberg, and S.~Vorobyov}, {\em On fixed-parameter
  complexity of infinite games}, in The Nordic Workshop on Programming Theory
  (NWPT 2003), vol.~34, Citeseer, 2003, pp.~29--31.

\bibitem{calude2017deciding}
{\sc C.~S. Calude, S.~Jain, B.~Khoussainov, W.~Li, and F.~Stephan}, {\em
  Deciding parity games in quasipolynomial time}, in Proceedings of the 49th
  Annual ACM SIGACT Symposium on Theory of Computing, 2017, pp.~252--263.
\newblock STOC 2017 Best Paper Award.

\bibitem{casares2024simple}
{\sc A.~Casares, M.~Pilipczuk, M.~Pilipczuk, U.~S. Souza, and K.~Thejaswini},
  {\em Simple and tight complexity lower bounds for solving rabin games}, in
  2024 Symposium on Simplicity in Algorithms (SOSA), SIAM, 2024, pp.~160--167.

\bibitem{daviaud2020strahler}
{\sc L.~Daviaud, M.~Jurdzi{\'n}ski, and K.~Thejaswini}, {\em The strahler
  number of a parity game}, in 47th International Colloquium on Automata,
  Languages, and Programming, ICALP 2020, Schloss Dagstuhl-Leibniz-Zentrum fur
  Informatik GmbH, Dagstuhl Publishing, 2020, p.~123.

\bibitem{dell2022smaller}
{\sc D.~Dell'Erba and S.~Schewe}, {\em Smaller progress measures and separating
  automata for parity games}, Frontiers in Computer Science, 4 (2022),
  p.~936903.

\bibitem{dziembowski1997much}
{\sc S.~Dziembowski, M.~Jurdzinski, and I.~Walukiewicz}, {\em How much memory
  is needed to win infinite games?}, in Proceedings of Twelfth Annual IEEE
  Symposium on Logic in Computer Science, IEEE, 1997, pp.~99--110.

\bibitem{emerson1988complexity}
{\sc E.~A. Emerson and C.~S. Jutla}, {\em The complexity of tree automata and
  logics of programs}, in FoCS, vol.~88, 1988, pp.~328--337.

\bibitem{emerson1991tree}
{\sc E.~A. Emerson and C.~S. Jutla}, {\em Tree automata, mu-calculus and
  determinacy}, in FoCS, vol.~91, Citeseer, 1991, pp.~368--377.

\bibitem{emerson1999complexity}
{\sc E.~A. Emerson and C.~S. Jutla}, {\em The complexity of tree automata and
  logics of programs}, SIAM Journal on Computing, 29 (1999), pp.~132--158.

\bibitem{fearnley2017ordered}
{\sc J.~Fearnley, S.~Jain, S.~Schewe, F.~Stephan, and D.~Wojtczak}, {\em An
  ordered approach to solving parity games in quasi polynomial time and quasi
  linear space}, in Proceedings of the 24th ACM SIGSOFT International SPIN
  Symposium on Model Checking of Software, 2017, pp.~112--121.

\bibitem{fijalkow2023games}
{\sc N.~Fijalkow, N.~Bertrand, P.~Bouyer-Decitre, R.~Brenguier, A.~Carayol,
  J.~Fearnley, H.~Gimbert, F.~Horn, R.~Ibsen-Jensen, N.~Markey, B.~Monmege,
  P.~Novotný, M.~Randour, O.~Sankur, S.~Schmitz, O.~Serre, and M.~Skomra},
  {\em Games on graphs}, 2023, \url{https://arxiv.org/abs/2305.10546}.
\newblock To be published by Cambridge University Press. Editor: Nathanaël
  Fijalkow.

\bibitem{gradel2002automata}
{\sc E.~Gr{\"a}del, W.~Thomas, and T.~Wilke}, {\em Automata, logics, and
  infinite games. lncs, vol. 2500}, 2002.

\bibitem{horn2005streett}
{\sc F.~Horn}, {\em Streett games on finite graphs}, in Proc. 2nd Workshop
  Games in Design Verification (GDV), Citeseer, 2005.

\bibitem{horn2008explicit}
{\sc F.~Horn}, {\em Explicit muller games are ptime}, in IARCS Annual
  Conference on Foundations of Software Technology and Theoretical Computer
  Science, Schloss Dagstuhl-Leibniz-Zentrum f{\"u}r Informatik, 2008.

\bibitem{hunter2008complexity}
{\sc P.~Hunter and A.~Dawar}, {\em Complexity bounds for muller games},
  Theoretical Computer Science (TCS),  (2008).

\bibitem{jurdzinski2017succinct}
{\sc M.~Jurdzi{\'n}ski and R.~Lazi{\'c}}, {\em Succinct progress measures for
  solving parity games}, in 2017 32nd Annual ACM/IEEE Symposium on Logic in
  Computer Science (LICS), IEEE, 2017, pp.~1--9.

\bibitem{liang2023connectivity}
{\sc Z.~Liang, B.~Khoussainov, T.~Takisaka, and M.~Xiao}, {\em Connectivity in
  the presence of an opponent}, in 31st Annual European Symposium on Algorithms
  (ESA 2023), Schloss Dagstuhl-Leibniz-Zentrum f{\"u}r Informatik, 2023.

\bibitem{majumdar2024rabin}
{\sc R.~Majumdar, I.~Sa{\u{g}}lam, and K.~Thejaswini}, {\em Rabin games and
  colourful universal trees}, in International Conference on Tools and
  Algorithms for the Construction and Analysis of Systems, Springer, 2024,
  pp.~213--231.

\bibitem{mcnaughton1993infinite}
{\sc R.~McNaughton}, {\em Infinite games played on finite graphs}, Annals of
  Pure and Applied Logic, 65 (1993), pp.~149--184.

\bibitem{neider2014down}
{\sc D.~Neider, R.~Rabinovich, and M.~Zimmermann}, {\em Down the borel
  hierarchy: Solving muller games via safety games}, Theoretical Computer
  Science, 560 (2014), pp.~219--234.

\bibitem{nerode1996mcnaughton}
{\sc A.~Nerode, J.~B. Remmel, and A.~Yakhnis}, {\em Mcnaughton games and
  extracting strategies for concurrent programs}, Annals of Pure and Applied
  Logic, 78 (1996), pp.~203--242.

\bibitem{piterman2006faster}
{\sc N.~Piterman and A.~Pnueli}, {\em Faster solutions of rabin and streett
  games}, in 21st Annual IEEE Symposium on Logic in Computer Science (LICS'06),
  IEEE, 2006, pp.~275--284.

\bibitem{zielonka1998infinite}
{\sc W.~Zielonka}, {\em Infinite games on finitely coloured graphs with
  applications to automata on infinite trees}, Theoretical Computer Science,
  200 (1998), pp.~135--183.

\end{thebibliography}
